\documentclass[review]{elsarticle}

\usepackage{mathrsfs,amsmath}
\usepackage{graphicx}
\usepackage{subfigure}
\usepackage{amssymb}
\usepackage{graphicx}
\usepackage{amsfonts}
\usepackage{amsthm}
\biboptions{numbers,sort&compress}
\newtheorem{thm}{Theorem}
\newtheorem{defn}{Definition}

\begin{document}

\begin{frontmatter}



\title{The two-dimensional OLCT of angularly periodic functions in polar coordinates}


\author{Hui Zhao$^{a,b}$}
\author{Bing-Zhao Li$^{a,b}$\corref{mycorrespondingauthor}}
\cortext[mycorrespondingauthor]{Corresponding author}\ead{li\_bingzhao@bit.edu.cn}

\address{$^{a}$School of Mathematics and Statistics, Beijing Institute of Technology, Beijing 100081, China}
\address{$^{b}$Beijing Key Laboratory on MCAACI, Beijing Institute of Technology, Beijing 100081, China}

\begin{abstract}
 The two-dimensional (2D) offset linear canonical transform (OLCT) in polar coordinates plays an important role in many fields of optics and signal processing. This paper studies the 2D OLCT in polar coordinates.
Firstly, we extend the 2D OLCT to the polar coordinate system, and obtain the offset linear canonical Hankel transform (OLCHT) formula. Secondly, through the angular periodic function with a period of $2\pi$, the relationship between the 2D OLCT and the OLCHT is revealed.
Finally, the spatial shift and convolution theorems for the 2D OLCT are proposed by using this relationship.
\end{abstract}


\end{frontmatter}


\section{Introduction}
In recent years, the linear canonical transform (LCT) \cite{[1],[2],[3],[4],[5],[6]} has become a new and more extensive signal analysis tool in the fields of applied mathematics, signal processing and optical system analysis \cite{[3],[4],[5]}. Compared with the LCT, the offset linear canonical transform (OLCT) \cite{[7],[8],[9]} adds two additional parameters $\tau$ and $\eta$ (corresponding to time offset and frequency modulation, respectively) on its basis, which is more versatile and flexible than the original LCT. The OLCT, also known as special affine Fourier transform \cite{[10]} or non-homogeneous regular transform \cite{[7]}, is a linear integral transform with six parameters $(a, b, c, d, \tau,\eta)$. It is a time-shifted and frequency-modulated version of LCT \cite{[11],[12],[13]}. Many linear transforms widely used in practical applications, such as Fourier transform (FT), offset FT \cite{[7],[9]}, fractional Fourier transform (FRFT) \cite{[7],[9],[14]}, offset FRFT \cite{[12],[15]}, Fresnel transform (FRST) \cite{[2]}, LCT, etc., are all special cases of the OLCT. Therefore, studying the OLCT and developing related theories of the OLCT may help to gain more insights into its special situation and transfer the knowledge gained from one discipline to other disciplines.\

As an extension of many other linear transform, the OLCT has a wide range of applications in optics and signal processing. For the OLCT domain signal, people have carried out a lot of promotion, and got a lot of research results on the OLCT convolution, sampling, etc., which can be found in \cite{[16],[17],[18],[19],[20]}. All these results are signals for dealing with one-dimensional (1D) problems. As we all know, in addition to the usual Cartesian coordinates, polar coordinates can also be used to represent signals. This situation is particularly suitable for the transformation of functions that are naturally described by polar coordinates, such as photoacoustics \cite{[21]}, computed tomography \cite{[22]}, and magnetic resonance imaging \cite{[23]}. The situation of two-dimensional (2D) OLCT in polar coordinates is still unknown. Therefore, it is very important to understand whether the results of existing one-dimensional case can be extended to the two-dimensional polar coordinate case.\

The two-dimensional polar coordinate transform means that the general function in the Cartesian coordinate system produces a special function, which forms an angular period with a period of $2\pi$ in the polar coordinate system, and vice versa. Therefore, the purpose of this article is to obtain a two-dimensional correlation study of the specific function class in the polar coordinate system. First, a new representation of the 2D OLCT in polar coordinates is studied, and the expression of offset linear canonical Hankel transform (OLCHT) is derived in detail. Second, using an extension of the celebrated exponent expansion formula \cite{[24]}, the relationship between the 2D OLCT and the OLCHT is obtained. Third, based on this relationship, the spatial shift theorem and convolution theorem of the 2D OLCT are studied. The results of this paper not only study some useful properties of the 2D OLCT in polar coordinates, but also provide a theoretical basis for the practical applications of optics and signal processing.\

The paper is organized as follows. Section \ref{Preli} provides some basic knowledge of the 2D OLCT in Cartesian coordinate system. In Section \ref{2D in polar cooddinates}, the definition of the 2D OLCT and the OLCHT in polar coordinates is given. In Section \ref{2D Relationship}, the relationship between the 2D OLCT and the OLCHT is discussed. Section \ref{The properties} proves the space shift theorem and convolution theorem of the 2D OLCT. Section \ref{Con} concludes the article.

\section{Preliminaries}
\label{Preli}

In this section, we give some necessary background and nation on the 2D OLCT.\

Let $\mathbf{t}=(t_{1},t_{2})$, $\mathbf{u}=(u_{1},u_{2})$ and $\mathbf{t}\cdot \mathbf{u}=t_{1}u_{1}+t_{2}u_{2}$. The 2D OLCT with real parameters of $A=(a,b,c,d,\tau,\eta)$ of a signal $f(\mathbf{t})$ is defined by \cite{[7],[8],[9],[18]}
\begin{align}
	\begin{split}
		F^{A}(\mathbf{u})=O^{A}_{L}[f(\mathbf{t})](\mathbf{u})=\begin{cases}
			\int_{-\infty}^{+\infty}\int_{-\infty}^{+\infty}f(\mathbf{t})h_{A}(\mathbf{t},\mathbf{u})\rm{d}\mathbf{t},   &b\neq0  \\
			\sqrt{d}e^{i\left[ \frac{cd}{2}(\mathbf{u}-\tau)^{2}+\eta \mathbf{u}\right] }f\left[ d(\mathbf{u}-\tau)\right] ,    &b=0
		\end{cases}
	\end{split}
\end{align}
where the kernel $h_{A}(\mathbf{t},\mathbf{u})$ is
\begin{align}	
	h_{A}(\mathbf{t},\mathbf{u})=K_{A}\,e^{i\left[ \frac{a}{2b}|\mathbf{t}|^{2}+\frac{1}{b}\mathbf{t}(\tau-\mathbf{u})-\frac{1}{b}\mathbf{u}(d\tau-b\eta)+\frac{d}{2b}|\mathbf{u}|^{2}\right] },
\end{align}
\begin{align}
	K_{A}=\dfrac{1}{2\pi b}e^{i\frac{d}{2b}\tau^{2}},
\end{align}
where $ad-bc=1$, $O^{A}_{L}$ denotes the OLCT operator, $|\mathbf{t}|^{2}=t_{1}^{2}+t_{2}^{2}$, $|\mathbf{u}|^{2}=u_{1}^{2}+u_{2}^{2}$ and $\rm{d}\mathbf{t}=\rm{d}t_{1}\rm{d}t_{1}$.

The inverse of the 2D OLCT with real parameters $A=(a,b,c,d,\tau,\eta)$ is given by the 2D OLCT with parameters $A^{-1}=(d,-b,-c,a,b\eta-d\tau,c\tau-a\eta)$. The exact inverse OLCT expression is give by \cite{[7],[8],[9],[18]}
\begin{align}
	f(\mathbf{t})=O^{A^{-1}}_{L}[F^{A}(\mathbf{u})](\mathbf{t})=C\int_{-\infty}^{+\infty}F^{A}(\mathbf{u})h_{A^{-1}}(\mathbf{u},\mathbf{t})\rm{d}\mathbf{\mathbf{u}},
\end{align}
where $C=e^{i\frac{1}{2}\left[ cd\tau^{2}-2ad\tau\eta+ab\eta^{2}\right]}.$

The definition for case $b = 0$ is the limit of the integral in (1) for the case $b\neq0$ as $|b|\rightarrow0$, the OLCT is simply a time scaled version off multiplied by a linear chirp. Therefore, from now on we shall confine our attention to the 2D OLCT for $b\neq0$. And without loss of generality, we assume $b>0$ in the following sections. Some of the special cases of the OLCT are listed in Table 1.

\begin{table*}
	\centering
	\caption{Some of the specific cases of the 2D OLCT}
	\label{tab:1}  
	\begin{tabular}{l l}
		\hline\hline\noalign{\smallskip}	
		Transform & Parameters $A$  \\
		\noalign{\smallskip}\hline\noalign{\smallskip}
		$A=(a,b,c,d,\tau,\eta)$ & Offset linear canonical transform (OLCT)   \\
		$A=(a,b,c,d,0,0)$ & Linear canonical transform (LCT)  \\
		$A=(cos\theta,sin\theta,-sin\theta,cos\theta,0,0)$ & Fractional Fourier transform (FRFT)  \\
		$A=(0,1,-1,0,0,0)$ & Fourier transform (FT)  \\
		$A=(cos\theta,sin\theta,-sin\theta,cos\theta,\tau,\eta)$ & Offset fractional Fourier transform (OFRFT)  \\
		$A=(1,b,0,1,0,0)$ & Fresnel transform (FRST)  \\
		$A=(1,0,0,1,0,\eta)$ & Frequency modulation  \\
		$A=(d^{-1},0,0,d,0,0)$ & Time scaling  \\
		$A=(1,0,0,1,\tau,0)$ & Time shifting  \\
		\noalign{\smallskip}\hline
	\end{tabular}
\end{table*}

\section{The 2D offset linear canonical transform and offset linear canonical Hankel transform in polar cooddinates}
\label{2D in polar cooddinates}

\subsection{2D Offset linear canonical transform}
\label{OLCT}

	Here we make the following symbolic regulations: the function $f\left( r,\theta \right)$, its 2D FT $F\left( \rho,\phi \right)$ and the 2D OLCT $F^{A}\left( \rho,\phi \right)$ are angularly periodic with period $2\pi$ as they are essentially in the form
	$$f\left( rcos\theta,rsin\theta \right)\stackrel{\bigtriangleup}{=}f\left( r,\theta \right),$$
	$$F\left( \rho cos\phi,\rho sin\theta \right)\stackrel{\bigtriangleup}{=}F\left( \rho,\phi \right),$$
	$$F^{M}\left( \rho cos\phi,\rho sin\theta \right)\stackrel{\bigtriangleup}{=}F^{M}\left( \rho,\phi \right).$$
	By changing Cartesian coordinates to polar coordinates related to general functions $f\left( \mathbf{t} \right)$, its 2D FT $F\left( \mathbf{u} \right)$ and 2D OLCT $F^{A}\left( \mathbf{u} \right)$, respectively.
	
	According to the 2D OLCT, we can obtain the definition of the 2D OLCT in polar coordinates as follows:
	\begin{defn}[2D OLCT] 
	Let polar coordinates $t_{1}=rcos\theta$, $t_{2}=rsin\theta$, $u_{1}=\rho cos\phi$, $u_{2}=\rho sin\phi$. Assume the function $f(r,\theta)$ is angularly periodic in $2\pi$, then its 2D OLCT with real parameters of $A=(a,b,c,d,\tau,\eta)$ in polar coordinates is defined by   
	\begin{align}
		\begin{split}
			F^{A}(\rho,\phi)=O^{A}_{L}[f](\rho,\phi)=
			\int_{0}^{+\infty}\int_{0}^{2\pi}f(r,\theta)P_{A}(r,\theta;\rho,\phi)r\rm{d}r\rm{d}\theta,
		\end{split}
	\end{align}	
	where the $P_{A}(r,\theta;\rho,\phi)$ denotes the 2D OLCT kernel in polar coordinates and is given by
	\begin{align}	 
		P_{A}(r,\theta;\rho,\phi)=\dfrac{\ell_{A}}{2\pi b}e^{i\left[\frac{a}{2b}r^{2}+\frac{\sqrt{2}\tau r}{b}sin(\theta+\frac{\pi}{4})-\frac{\sqrt{2}\rho(d\tau-b\eta)}{b}sin(\phi+\frac{\pi}{4})-\frac{r\rho}{b}cos(\theta-\phi)+\frac{d}{2b}\rho^{2}\right] },
	\end{align}
	with $b\neq0$ and $\ell_{A}=e^{i\left[\frac{\left(ad+1\right)\tau^{2}}{2ab}+\frac{\left(d\tau-b\eta\right)^{2}}{2bd}\right]}$.\
\end{defn}

\begin{proof} Let $\mathbf{t}=(t_{1},t_{2})$, $\mathbf{u}=(u_{1},u_{2})$, and the 2D OLCT of $f(\mathbf{t})$ in (1), we have
	\begin{align}
		\begin{split}
			&\;\quad F^{A}(\mathbf{u})=O^{A}_{L}[f(\mathbf{t})](\mathbf{u})\\
			&=\int_{-\infty}^{+\infty}\int_{-\infty}^{+\infty}f(\mathbf{t})K_{A}\,e^{\frac{i}{2b}\left[ a|\mathbf{t}|^{2}+2\mathbf{t}(\tau-\mathbf{u})-2\mathbf{u}(d\tau-b\eta)+d|\mathbf{u}|^{2}\right]}\rm{d}\mathbf{t}\\
			&=\int_{-\infty}^{+\infty}\int_{-\infty}^{+\infty}f(\mathbf{t})K_{A}\,e^{\frac{i}{2b}\left[ \left(\sqrt{a}\mathbf{t}+\frac{\tau}{\sqrt{a}}\right) ^{2}-\frac{\tau^{2}}{a}+\left(\sqrt{d}\mathbf{u}-\frac{d\tau-b\eta}{\sqrt{d}}\right) ^{2}-\frac{\left( d\tau-b\eta\right) ^{2}}{d}-2\mathbf{t}\cdot\mathbf{u}\right]}\rm{d}\mathbf{t}	\\
			&=\int_{-\infty}^{+\infty}\int_{-\infty}^{+\infty}f(\mathbf{t})K_{A}\,e^{\frac{i}{2b}\left[a\left( t_{1}^{2}+t_{2}^{2}\right)+2\tau\left( t_{1}+t_{2}\right)+\frac{\tau^{2}}{a}+d\left( u_{1}^{2}+u_{2}^{2}\right) \right] }\\
			&\times e^{\frac{i}{2b}\left[-2\left(d\tau-b\eta\right)\left( u_{1}+u_{2}\right)+\frac{\left(d\tau-b\eta\right)^{2}}{d}-2\left(t_{1}u_{1}+t_{2}u_{2}\right)\right]}\rm{d}t_{1}\rm{d}t_{2},
		\end{split}
	\end{align}	
	where $K_{A}(t,u)$ is given by (3).\\
	Using polar coordinates $t_{1}=rcos\theta$, $t_{2}=rsin\theta$, $u_{1}=\rho cos\phi$, $u_{2}=\rho sin\phi$, we obtain
	\begin{align}
		\begin{split}
			&\;\quad F^{A}(\rho,\phi)=O^{A}_{L}[f](\rho,\phi)\\
			&=\int_{0}^{+\infty}\int_{0}^{2\pi}f(r,\theta)K_{A}\,e^{\frac{i}{2b}\left[ar^{2}+2\sqrt{2}\tau rsin(\theta+\frac{\pi}{4})+\frac{\tau^{2}}{a}+d\rho^{2}\right]}\\
			&\times e^{\frac{i}{2b}\left[-2\sqrt{2}\rho(d\tau-b\eta)sin(\phi+\frac{\pi}{4})+\frac{\left(d\tau-b\eta\right)^{2}}{d}-2r\rho cos\left(\theta-\phi\right)\right] }r\rm{d}r\rm{d}\theta\\
			&=\int_{0}^{+\infty}\int_{0}^{2\pi}f(r,\theta)e^{\frac{i}{2b}\left[ar^{2}+2\sqrt{2}\tau rsin(\theta+\frac{\pi}{4})+d\rho^{2}-2\sqrt{2}\rho(d\tau-b\eta)sin(\phi+\frac{\pi}{4})\right]}\\
			&\times e^{\frac{i}{2b}\left[-2r\rho cos\left(\theta-\phi\right)\right] }\dfrac{1}{2\pi b}e^{i\frac{d}{2b}\tau^{2}}e^{\frac{i}{2b}\left[ \frac{\tau^{2}}{a}+\frac{\left(d\tau-b\eta\right)^{2}}{d}\right] }r\rm{d}r\rm{d}\theta.
		\end{split}  
	\end{align}
	Hence
	\begin{align}
		\begin{split}
			F^{A}(\rho,\phi)=O^{A}_{L}[f](\rho,\phi)=
			\int_{0}^{+\infty}\int_{0}^{2\pi}f(r,\theta)P_{A}(r,\theta;\rho,\phi)r\rm{d}r\rm{d}\theta,
		\end{split}  
	\end{align}	
	where
	\begin{align}
		\begin{split}
			P_{A}(r,\theta;\rho,\phi)&=e^{i\left[\frac{a}{2b}r^{2}+\frac{\sqrt{2}\tau r}{b}sin(\theta+\frac{\pi}{4})-\frac{\sqrt{2}\rho(d\tau-b\eta)}{b}sin(\phi+\frac{\pi}{4})-\frac{r\rho}{b}cos(\theta-\phi)+\frac{d}{2b}\rho^{2}\right] }\\
			&\times 
			\dfrac{1}{2\pi b}e^{i\left[\frac{\left(ad+1\right)\tau^{2}}{2ab}+\frac{\left(d\tau-b\eta\right)^{2}}{2bd}\right]},\\
		\end{split} 
	\end{align}	
	which completes the proof.
\end{proof}	

The inversion formula of the 2D OLCT in polar coordinates takes
\begin{align}
	\begin{split}
		f(r,\theta)=O^{A^{-1}}_{L}[F^{A}](r,\theta),
	\end{split}
\end{align}	
where $A^{-1}=(d,-b,-c,a,b\eta-d\tau,c\tau-a\eta)$ and $b\neq0$.\

As it is seen, when $A=(0,1,-1,0,0,0)$, the 2D OLCT can be reduced to the 2D FT
\begin{align}
	\begin{split}
		F(\rho,\phi)&=\mathcal{F}[f](\rho,\phi)\\
		&=\frac{1}{2\pi}\int_{0}^{+\infty}\int_{0}^{2\pi}f(r,\theta)e^{-ir\rho cos\left( \theta-\phi\right)}r\rm{d}r\rm{d}\theta.
	\end{split}
\end{align}	

It follows that there is a relation between the 2D OLCT and the 2D FT
\begin{align}
	\begin{split}
		F^{A}(\rho,\phi)=\frac{\ell_{A}}{b}e^{i\left[ \frac{d}{2b}\rho^{2}-\frac{\sqrt{2}\rho\left( d\tau-b\eta\right)}{b}sin\left(\phi+\frac{\pi}{4} \right)\right] } F[\tilde{f}]\left( \frac{\rho}{b},\phi\right), 
	\end{split} 
\end{align}	
where $b\neq0$, $\ell_{A}$ is given by (6) and $\tilde{f}(r,\theta)=e^{i\left[ \frac{d}{2b}r^{2}+\frac{\sqrt{2}\tau r}{b}sin(\theta+\frac{\pi}{4})\right] }f\left(r,\theta \right)$.

\subsection{Offset linear canonical Hankel transform}
\label{OLCHT}
In this subsection, we discuss definition of the 2D OLCHT in polar cooddinates. Inspired by the literature \cite{[25],[26],[27]}, under the premise that the transformation function has rotational symmetry, the OLCHT is obtained from the OLCT. 
\begin{defn}[OLCHT] 
	The $n$th-order OLCHT of the real parameters matrix of $A=(a,b,c,d,\tau,\eta)$ is defined by   
	\begin{align}
		\begin{split}
			H_{n}^{A}[f](\rho)=i^{n}\frac{w_{1}\ell_{A}}{b}e^{i\frac{d}{2b}\rho^{2}}\int_{0}^{+\infty}w_{2}e^{i\frac{a}{2b}r^{2}}J_{n}\left(\frac{r\rho}{b}\right)f(r)r\rm{d}r,
		\end{split}
	\end{align}	
	where the $J_{n}$ is the $n$th-order Bessel function of the first kind and order $n\geq-\frac{1}{2}$, $b\neq0$, $\ell_{A}$ is given by (6), and 
	\begin{align}
		\begin{split}
			w_{1}=\sum_{m=-\infty}^{+\infty}J_{m}\left( \frac{\sqrt{2}\rho\left(d\tau-b\eta \right) }{b}\right),
			w_{2}=\sum_{m=-\infty}^{+\infty}J_{m}\left( \frac{\sqrt{2}\tau r}{b}\right).
		\end{split}
	\end{align}	
\end{defn}

\begin{proof} From (5), it follows that
	\begin{align}
		\begin{split}
			F^{A}(\rho,\phi)&=\dfrac{\ell_{A}}{2\pi b}\int_{0}^{+\infty}\int_{0}^{2\pi}f(r,\theta)
			e^{i\left[\frac{a}{2b}r^{2}+\frac{\sqrt{2}\tau r}{b}sin(\theta+\frac{\pi}{4})\right] }\\
			&\times e^{i\left[ -\frac{\sqrt{2}\rho(d\tau-b\eta)}{b}sin(\phi+\frac{\pi}{4})-\frac{r\rho}{b}cos(\theta-\phi)+\frac{d}{2b}\rho^{2}\right] }r\rm{d}r\rm{d}\theta,    
		\end{split}
	\end{align}	
	where $\ell_{A}$ is given by (6).\\ 
	In view of the relation \cite{[28]}
	\begin{align}
		\begin{split}
			e^{-itsin\theta}=\sum_{m=-\infty}^{+\infty}J_{m}\left(t\right)e^{-im\theta}.
		\end{split}
	\end{align}	
	So we can obtain
	\begin{align}
		\begin{split}
			e^{i\frac{\sqrt{2}\tau r}{b}sin(\theta+\frac{\pi}{4})}=\sum_{m=-\infty}^{+\infty}J_{m}\left(\frac{\sqrt{2}\tau r}{b}\right)e^{im\left( \theta+\frac{\pi}{4}\right) },
		\end{split}
	\end{align}	
	\begin{align}
		\begin{split}
			e^{-i\frac{\sqrt{2}\rho(d\tau-b\eta)}{b}sin(\phi+\frac{\pi}{4})}=\sum_{m=-\infty}^{+\infty}J_{m}\left(\frac{\sqrt{2}\rho(d\tau-b\eta)}{b}\right)e^{-im\left( \phi+\frac{\pi}{4}\right) }.
		\end{split}
	\end{align}	
	When $f\left(r,\theta \right) $ is circularly symmetric, there is
	\begin{align}
		\begin{split}
			f\left(r,\theta \right)=f\left(r\right)e^{ik\theta}.
		\end{split}
	\end{align}	
	Let 
	\begin{align}
		\begin{split}
			B=\int_{0}^{2\pi}f(r,\theta)e^{i\left[\frac{\sqrt{2}\tau r}{b}sin(\theta+\frac{\pi}{4})-\frac{\sqrt{2}\rho(d\tau-b\eta)}{b}sin(\phi+\frac{\pi}{4})-\frac{r\rho}{b}cos(\theta-\phi)\right] }\rm{d}\theta.
		\end{split}
	\end{align}	
	Using (18), (19) and (20), we can derive the following result
	\begin{align}
		\begin{split}
			B&=\int_{0}^{2\pi}\sum_{m=-\infty}^{+\infty}\sum_{m=-\infty}^{+\infty}J_{m}\left(\frac{\sqrt{2}\tau r}{b}\right)J_{m}\left(\frac{\sqrt{2}\rho(d\tau-b\eta)}{b}\right)\\
			&\times e^{im\theta-im\phi-i\frac{r\rho}{b}cos(\theta-\phi)}e^{ik\theta}f(r)\rm{d}\theta.
		\end{split}
	\end{align}	
	According to the famous formula \cite{[24],[28]}
	\begin{align}
		\begin{split}
			J_{n}(x)=\frac{1}{2\pi}\int_{0}^{2\pi}e^{i\left(n\theta-xsin\theta\right) }\rm{d}\theta.
		\end{split}
	\end{align}	
	We have
	\begin{align}
		\begin{split}
			&\quad \;\int_{0}^{2\pi}e^{im\theta-im\phi-i\frac{r\rho}{b}cos(\theta-\phi)}\cdot e^{ik\theta}\rm{d}\theta\\
			&=\int_{0}^{2\pi}e^{im\theta-im\phi-i\frac{r\rho}{b}sin(\frac{\pi}{2}+\theta-\phi)}\cdot e^{ik\theta}\rm{d}\theta\\
			&=e^{-in(\frac{\pi}{2}+\theta-\phi)+im\theta-im\phi+ik\theta}\int_{0}^{2\pi}e^{i\left[ n(\frac{\pi}{2}+\theta-\phi)-\frac{r\rho}{b}sin(\frac{\pi}{2}+\theta-\phi)\right] }\rm{d}\theta\\
			&=2\pi J_{n}\left( \frac{r\rho}{b}\right) e^{i\left[ (m+k-n)\theta+i(n-m)\phi-i\frac{\pi}{2}n\right]}.
		\end{split}
	\end{align}	
	By making the change of $k=n-m$ in the above expression, we obtain
	\begin{align}
		\begin{split}
			\int_{0}^{2\pi}e^{im\theta-im\phi-i\frac{r\rho}{b}cos(\theta-\phi)}\cdot e^{ik\theta}\rm{d}\theta&=2\pi J_{n}\left( \frac{r\rho}{b}\right)e^{ik\phi-i\frac{\pi}{2}n}.
		\end{split}
	\end{align}	
	Hence
	\begin{align}
		\begin{split}
			F^{A}\left(\rho,\phi\right)&=i^{n}\dfrac{\ell_{A}}{2\pi b}\sum_{m=-\infty}^{+\infty}J_{m}\left(\frac{\sqrt{2}\rho(d\tau-b\eta)}{b}\right)e^{i\frac{d}{2b}\rho^{2}}\\
			&\times\int_{0}^{+\infty}\sum_{m=-\infty}^{+\infty}J_{m}\left(\frac{\sqrt{2}\tau r}{b}\right)e^{i\frac{a}{2b}r^{2}}2\pi J_{n}\left( \frac{r\rho}{b}\right)e^{ik\phi}r\rm{d}r.
		\end{split}
	\end{align}	
	The output function also has circular symmetry, we have 
	\begin{align}
		\begin{split}
			F^{A}\left(\rho,\phi\right)=F^{A}\left(\rho\right) e^{ik\phi}.
		\end{split}
	\end{align}	
	So
	\begin{align}
		\begin{split}
			&\;\quad H_{n}^{A}[f](\rho)=F^{A}\left(\rho\right)\\
			&=i^{n}\dfrac{\ell_{A}}{b}\sum_{m=-\infty}^{+\infty}J_{m}\left(\frac{\sqrt{2}\rho(d\tau-b\eta)}{b}\right)e^{i\frac{d}{2b}\rho^{2}}\\
			&\times\int_{0}^{+\infty}\sum_{m=-\infty}^{+\infty}J_{m}\left(\frac{\sqrt{2}\tau r}{b}\right)e^{i\frac{a}{2b}r^{2}}J_{n}\left( \frac{r\rho}{b}\right)r\rm{d}r.
		\end{split}
	\end{align}	
The proof is completed.
\end{proof}

The inversion formula of $n$th-order OLCHT takes
\begin{align}
	\begin{split}
		f(r)=H_{n}^{-A^{-1}}\left[ H_{n}^{A}\left[ f\right] \right] (r).
	\end{split}
\end{align}
It is evident that the $n$th-order OLCHT and its inverse with $A=(0,1,-1,0,0,0)$ reduces to the conventional $n$th-order Hankel transform (HT) \cite{[29]}
\begin{align}
	\begin{split}
		H_{n}[f](\rho)=\int_{0}^{+\infty}f(r)J_{n}\left(\rho r\right)r\rm{d}r,
	\end{split}
\end{align}	  
and the corresponding inversion formula
\begin{align}
	\begin{split}
		f(r)=H_{n}\left[ H_{n}\left[ f\right] \right] (r),
	\end{split}
\end{align}
respectively.\

\section{Relationship between the 2D OLCT and the OLCHT}
\label{2D Relationship}
Under the condition of rotational symmetry, the OLCHT can be deduced from the OLCT, and there is a close relationship between them.

\subsection{Relationship between the 2D OLCHT and the HT}
\label{HT}
The purpose of this section is to extend the above relationship to preparing for the OLCT domain. Let the function $f(r,\theta)$ and its 2D FT  $\left(F(\rho,\phi)\right)$ satisfy the Dirichlet conditions. Since they can angularly periodic with period $2\pi$, then their Fouries series are well-defined as
\begin{align}
	\begin{split}
		f(r,\theta)=\sum_{n=-\infty}^{+\infty}f_{n}\left( r\right) e^{in\theta},
	\end{split}
\end{align}	 
\begin{align}
	\begin{split}
		F(\rho,\phi)=\sum_{n=-\infty}^{+\infty}F_{n}\left( r\right) e^{in\phi}.
	\end{split}
\end{align}	  

As we all know, the $n$th term in Fourier series of primitive functions and their 2D FT versions $f_{n}(r)$ and $F_{n}(\rho)$, forming an $n$th-order HT pair \cite{[29],[30],[31]}
\begin{align}
	\begin{split}
		F_{n}(\rho)=i^{-n}H_{n}[f_{n}](\rho),
	\end{split}
\end{align}	
\begin{align}
	\begin{split}
		f_{n}(r)=i^{n}H_{n}[F_{n}](r).
	\end{split}
\end{align}	

From the above, (14) and (30), we can obtain the relationship between the OLCHT and the HT, it follows that 
\begin{align}
	\begin{split}
		H_{n}^{A}[f](\rho)=i^{n}\frac{w_{1}\ell_{A}}{b}e^{i\frac{d}{2b}\rho^{2}}H_{n}[\tilde{f}]\left( \frac{\rho}{b}\right),
	\end{split}
\end{align}	
where $w_{1}, w_{2}$ is given by (15), and $\ell_{A}$ is given by (6)
\begin{align}
	\begin{split}
		\tilde{f}(r)=w_{2}e^{i\frac{a}{2b}r^{2}}f(r). 
	\end{split}
\end{align}	
\textbf{Remark 1.}
If it is assumed that $f$ is radially symmetric, then it can be written as a function of $r$ only and can thus be taken out of the integration over the angular coordinate.\\
\textbf{Remark 2.}
 When the function $f(r,\theta)$ is not radially symmetric and is a function of both $r$ and $\theta$, the preceding result can be generalized. Since $f(r,\theta)$ depends on the angle $\theta$, it can be expanded into a Fourier series (32) and (33).
\subsection{Relationship between the 2D OLCT and the OLCHT}
\label{between}
Next, we will show that the relation between the 2D OLCT and the OLCHT.
\begin{thm}
	Let the function $f(r,\theta)$ satisfy Dirichlet conditions, be angularly periodic in $2\pi$, and have a Fourier expansion
	\begin{align}
		\begin{split}
			f(r,\theta)=\sum_{n=-\infty}^{+\infty}f_{n}\left( r\right) e^{in\theta},    
		\end{split}
	\end{align}
	then the Fourier series expansion of the 2D OLCT of $f(r,\theta)$ has a form 
	\begin{align}
		\begin{split}
			F^{A}(\rho,\phi)=\sum_{n=-\infty}^{+\infty}H_{n}^{A}\left[ f_{n}\right] \left(\rho\right) e^{in\phi}.    
		\end{split}
	\end{align} 
\end{thm}
\begin{proof}
	The 2D OLCT of $f(r,\theta)$ as expanded in (38), takes
	\begin{align}
		\begin{split}
			F^{A}(\rho,\phi)&=\sum_{n=-\infty}^{+\infty}\int_{0}^{+\infty}\int_{0}^{2\pi}e^{in\theta}f_{n}(r)P_{A}(r,\theta;\rho,\phi)r\rm{d}r\rm{d}\theta\\
			&=\sum_{n=-\infty}^{+\infty}\int_{0}^{+\infty}\underbrace{e^{i\left[ 	\frac{\sqrt{2}\tau r}{b}sin(\theta+\frac{\pi}{4})-\frac{\sqrt{2}\rho\left(d\tau-b\eta \right) }{b}sin(\phi+\frac{\pi}{4})-\frac{r\rho}{b} cos\left(\theta-\phi\right) \right]}e^{in\theta}}_{D(r,\theta;\rho,\phi)}\\
			&\times\frac{\ell_{A}}{2\pi b}e^{i\left(\frac{a}{2b}r^{2}+\frac{d}{2b}\rho^{2} \right) }\int_{0}^{2\pi} f_{n}(r) r\rm{d}r\rm{d}\theta.
		\end{split}
	\end{align}
	According to the celebrated exponent expansion formula \cite{[28]}
	\begin{align}
		\begin{split}
			e^{i\frac{r\rho}{b} cos\left(\theta-\phi\right)}=e^{i\frac{r\rho}{b} cos\left(\phi-\theta\right)}=\sum_{v=-\infty}^{+\infty}i^{v}J_{v}\left(\frac{r\rho}{b}\right)e^{-iv\left( \phi-\theta\right) }.
		\end{split}
	\end{align} 
	Using (18), (19) and (41), we get
	\begin{align}
		\begin{split}
			D(r,\theta;\rho,\phi)&=\sum_{v=-\infty}^{+\infty}i^{v}J_{v}\left(\frac{r\rho}{b}\right)w_{1}w_{2}e^{im\theta-im\phi-iv\left( \theta-\phi\right)+in\theta}\\
			&=\sum_{v=-\infty}^{+\infty}i^{v}J_{v}\left(\frac{r\rho}{b}\right)w_{1}w_{2}e^{i\left( m-v+n\right)\theta -i\left( m-v\right)\phi},
		\end{split}
	\end{align} 
	where $w_{1}, w_{2}$ is given by (15). \\
	For $n=v-m$, (42) can be written as
	\begin{align}
		\begin{split}
			D(r,\theta;\rho,\phi)=\sum_{v=-\infty}^{+\infty}i^{v}J_{v}\left(\frac{r\rho}{b}\right)w_{1}w_{2}e^{in\phi}.
		\end{split}
	\end{align} 
	Hence
	\begin{align}
		\begin{split}
			&\;\quad F^{A}(\rho,\phi)\\
			&=\sum_{n=-\infty}^{+\infty}\int_{0}^{+\infty}\frac{\ell_{A}}{2\pi b}e^{i\left(\frac{a}{2b}r^{2}+\frac{d}{2b}\rho^{2} \right) }\sum_{v=-\infty}^{+\infty}i^{v}J_{v}\left(\frac{r\rho}{b}\right)w_{1}w_{2}e^{in\phi}f_{n}(r)r\rm{d}r\int_{0}^{2\pi}\rm{d}\theta\\
			&=\sum_{n=-\infty}^{+\infty}\sum_{v=-\infty}^{+\infty}i^{v}\frac{w_{1}\ell_{A}}{b}e^{i\frac{d}{2b}\rho^{2}}\int_{0}^{+\infty}w_{2}e^{i\frac{a}{2b}r^{2}}J_{v}\left(\frac{r\rho}{b}\right)f_{n}(r)e^{in\phi}r\rm{d}r\\
			&=\sum_{n=-\infty}^{+\infty}i^{n}\frac{w_{1}\ell_{A}}{b}e^{i\frac{d}{2b}\rho^{2}}\int_{0}^{+\infty}w_{2}e^{i\frac{a}{2b}r^{2}}J_{n}\left(\frac{r\rho}{b}\right)f_{n}(r)e^{in\phi}r\rm{d}r\\
			&=\sum_{n=-\infty}^{+\infty}H_{n}^{A}\left[ f_{n}\right]\left(\rho \right)e^{in\phi}.  
		\end{split}
	\end{align}
	The proof is completed. 
\end{proof}
Through the above series of discussions, the Fourier series expansion of the 2D OLCT $F^{A}(\rho,\phi)$ has a form
\begin{align}
	\begin{split}
		F^{A}(\rho,\phi)=\sum_{n=-\infty}^{+\infty}F_{n}^{A}\left( r\right) e^{in\phi}.    
	\end{split}
\end{align}
According to Theorem 1 and (45), we get
\begin{align}
	\begin{split}
		F^{A}_{n}(\rho)=H_{n}^{A}\left[ f_{n}\right](\rho).   
	\end{split}
\end{align}

\section{The properties of 2D OLCT in polar cooddinates}
\label{The properties}

The 1D OLCT has many important properties, such as spatial shift and convolution. Similarly, the 2D OLCT also has some important nature. In this section, we use two theorems to understand the properties of the 2D OLCT in polar coordinates, namely spatial shift theorem and convolution theorem, respectively.  
\subsection{Spatial shift theorem}
\label{Spatial}

See \cite{[18]} for the spatial shift properties of the 1D OLCT. Below we give the spatial shift theorem of the 2D OLCT in polar coordinates.
\begin{thm}[Spatial shift theorem]
	Let $\mathbf{t_{0}}=(t_{3},t_{4})$, $t_{3}=r_{0}cos\theta_{0}$, $t_{4}=r_{0}sin\theta_{0}$, using polar coordinates. Then we have
	\begin{align}
		\begin{split}
			f\left(\mathbf{t}-\mathbf{t_{0}} \right) &=(-i)^{\frac{n}{2}}2\pi cw_{3}\gamma\left(r,\theta,r_{0},\theta_{0} \right)\sum_{n=-\infty}^{+\infty}\\
			&\times e^{in(\theta-\theta_{0})}\int_{0}^{+\infty}H_{n}^{A}\left[f_{n} \right] (\rho)e^{-i\frac{d}{2b}\left( \rho^{2}\right)}\\
			&\times J_{n}\left( \frac{r\rho}{b}\right) J_{n}\left( \frac{r_{0}\rho}{b}\right) J_{n}\left( \frac{\sqrt{2}\rho\left( b\eta-d\tau\right) }{b}\right) \rho\rm{d}\rho,
		\end{split}
	\end{align}	
	where
	$$\gamma\left(r,\theta,r_{0},\theta_{0} \right) = e^{i\left[ -\frac{a}{2b}\left( r^{2}+r_{0}^{2}\right) +\frac{arr_{0}}{b}cos(\theta-\theta_{0})+\frac{\sqrt{2}\tau r(bc-ad)}{b}sin(\theta+\frac{\pi}{4})-\frac{\sqrt{2}\tau r_{0}(bc-ad)}{b}sin(\theta_{0}+\frac{\pi}{4}) \right]},$$ 
	$$w_{3}=e^{-i\left[\frac{\left( b\eta-d\tau\right) ^{2}}{2bd} +\frac{\tau^{2}\left( bc-ad \right) ^{2}}{2ab} \right]}.$$ 
\end{thm}
\begin{proof}
	According to (11), it follows that
	\begin{align}
		\begin{split}
			f\left(\mathbf{t}-\mathbf{t_{0}} \right) &=C\int_{-\infty}^{+\infty}\int_{-\infty}^{+\infty}F^{A}(\mathbf{u})K_{A^{-1}}e^{i\left[-\frac{d}{2b}|\mathbf{u}|^{2}+\frac{1}{b}\mathbf{u}\left(\mathbf{t}-\mathbf{t_{0}}\right) \right]}\\ 
			&\times e^{i\left[ -\frac{1}{b}\mathbf{u}\left(b\eta-d\tau\right)+\frac{\left(\mathbf{t}-\mathbf{t_{0}} \right)\tau }{b}\left(bc-ad \right)-\frac{a}{2b}|\mathbf{t}-\mathbf{t_{0}}|^{2}\right] }\rm{d}\mathbf{u}.
		\end{split}
	\end{align}	
	Using polar coordinates $t_{1}=rcos\theta$, $t_{2}=rsin\theta$, $t_{3}=r_{0}cos\theta_{0}$, $t_{4}=r_{0}sin\theta_{0}$, $u_{1}=\rho cos\phi$, $u_{2}=\rho sin\phi$, we obtain 
	\begin{align}
		\begin{split}
			E&\stackrel{\triangle}{=}e^{-\frac{i}{2b}\left[d|\mathbf{u}|^{2}+2(b\eta-d\tau)\mathbf{u}+a|\mathbf{t}-\mathbf{t_{0}}|^{2}- 2\tau\left( \mathbf{t}-\mathbf{t_{0}}\right) (bc-ad)-2\mathbf{u}\mathbf{t}+2\mathbf{u}\mathbf{t_{0}}\right]}\\
			&=e^{-\frac{i}{2b}\left[d|\mathbf{u}|^{2}+2(b\eta-d\tau)\mathbf{u}+a|\mathbf{t}-\mathbf{t_{0}}|^{2}- 2\tau\left( \mathbf{t}-\mathbf{t_{0}}\right) (bc-ad)\right]}e^{\frac{i}{2b}\left[ 2\rho rcos\left( \theta-\phi\right) -2\rho r_{0}cos\left( \theta_{0}-\phi\right)\right] }\\
			&=e^{-\frac{i}{2b}\left[ \left(\sqrt{d}\mathbf{u}+\frac{b\eta-d\tau}{\sqrt{d}}\right) ^{2}-\frac{\left( b\eta-d\tau\right) ^{2}}{d}+	\left(\sqrt{a}\left( \mathbf{t}-\mathbf{t_{0}}\right) -\frac{\tau\left( bc-ad\right) }{\sqrt{a}}\right) ^{2}-\frac{\tau^{2}\left(bc-ad \right) ^{2}}{a}\right]}\\
			&\times e^{\frac{i}{2b}\left[ 2\rho rcos\left( \theta-\phi\right) -2\rho r_{0}cos\left( \theta_{0}-\phi\right)\right] }\\
			&=e^{i\left[ -
				\frac{d}{2b}\rho^{2} 	
				-\frac{a}{2b}\left( r^{2}+r_{0}^{2}\right) +\frac{arr_{0}}{b}cos(\theta-\theta_{0})-\frac{\sqrt{2}\rho (b\eta-d\tau)}{b}sin(\phi+\frac{\pi}{4})+\frac{\sqrt{2}r\tau (bc-ad)}{b}sin(\theta+\frac{\pi}{4}) \right] }\\
			&\times e^{i\left[-\frac{\sqrt{2}\tau r_{0}(bc-ad)}{b}sin(\theta_{0}+\frac{\pi}{4})-\frac{\left( b\eta-d\tau\right)^{2} }{2bd}-\frac{\tau^{2}\left( bc-ad\right)^{2} }{2ab} \right] }e^{\frac{i}{2b}\left[ 2\rho rcos\left( \theta-\phi\right) -2\rho r_{0}cos\left( \theta_{0}-\phi\right)\right] }\\
			&=e^{i\left[ -
				\frac{d}{2b}\rho^{2} 	
				-\frac{a}{2b}\left( r^{2}+r_{0}^{2}\right) +\frac{arr_{0}}{b}cos(\theta-\theta_{0})-\frac{\sqrt{2}\rho (b\eta-d\tau)}{b}sin(\phi+\frac{\pi}{4})+\frac{\sqrt{2}r\tau (bc-ad)}{b}sin(\theta+\frac{\pi}{4}) \right] }\\
			&\times e^{i\left[-\frac{\sqrt{2}\tau r_{0}(bc-ad)}{b}sin(\theta_{0}+\frac{\pi}{4}) \right] }\underbrace{e^{-i\left[ \frac{\left( b\eta-d\tau\right)^{2} }{2bd}+\frac{\tau^{2}\left( bc-ad\right)^{2} }{2ab}\right] }}_{w_{3}}e^{\frac{i}{2b}\left[ 2\rho rcos\left( \theta-\phi\right) -2\rho r_{0}cos\left( \theta_{0}-\phi\right)\right] }\\
			&=w_{3}e^{i\left[ -
				\frac{d}{2b}\rho^{2} 	
				-\frac{a}{2b}\left( r^{2}+r_{0}^{2}\right) +\frac{arr_{0}}{b}cos(\theta-\theta_{0})+\frac{\sqrt{2}r\tau (bc-ad)}{b}sin(\theta+\frac{\pi}{4}) -\frac{\sqrt{2}\tau r_{0}(bc-ad)}{b}sin(\theta_{0}+\frac{\pi}{4})\right] }\\
			&\times \underbrace{e^{\frac{i}{2b}\left[ -2\sqrt{2}\rho (b\eta-d\tau)sin(\phi+\frac{\pi}{4})+2\rho rcos\left( \theta-\phi\right) -2\rho r_{0}cos\left( \theta_{0}-\phi\right)\right] }}_{w_{4}}.
		\end{split}
	\end{align}	
	According to the celebrated exponent expansion formula (17), split each item of $w_{4}$, we have
	\begin{align}
		\begin{split}
			e^{i\frac{r\rho}{b} cos\left(\theta-\phi\right)}=\sum_{m=-\infty}^{+\infty}i^{m}J_{m}\left(\frac{r\rho}{b}\right)e^{im\left( \theta-\phi\right)},
		\end{split}
	\end{align}
	\begin{align}
		\begin{split}
			e^{-i\frac{r_{0}\rho}{b} cos\left(\theta_{0}-\phi\right)}=\sum_{m=-\infty}^{+\infty}i^{-m}J_{m}\left(\frac{r_{0}\rho}{b}\right)e^{-im\left( \theta_{0}-\phi\right)},
		\end{split}
	\end{align}	
and
	\begin{align}
		\begin{split}
			e^{-i\frac{\sqrt{2}\rho (b\eta-d\tau)}{b} sin(\phi+\frac{\pi}{4})}&=\sum_{n=-\infty}^{+\infty}J_{n}\left(\frac{\sqrt{2}\rho (b\eta-d\tau)}{b}\right)e^{-in\left(\phi+\frac{\pi}{4} \right)}.
		\end{split}
	\end{align}	
	Using (39), (49), (50), (51) and (52), we obtain
	\begin{align}
		\begin{split}
			&\;\quad f\left(\mathbf{t}-\mathbf{t_{0}} \right)\\ &=Cw_{3}K_{A^{-1}}\int_{0}^{+\infty}\sum_{m=-\infty}^{+\infty}J_{m}\left(\frac{r\rho}{b}\right)J_{m}\left(\frac{r_{0}\rho}{b}\right)     \sum_{n=-\infty}^{+\infty}J_{n}\left(\frac{\sqrt{2}\rho (b\eta-d\tau)}{b}\right) \\ &\times\sum_{n=-\infty}^{+\infty}H_{n}^{A}\left[ f_{n}\right] \left(\rho \right)e^{im\left( \theta-\theta_{0}\right)-in\frac{\pi}{4} }\int_{0}^{2\pi}\rm{d}\phi\rho\rm{d}\rho
			\\ &\times e^{i\left[ -
				\frac{d}{2b}\rho^{2} 	
				-\frac{a}{2b}\left( r^{2}+r_{0}^{2}\right) +\frac{arr_{0}}{b}cos(\theta-\theta_{0})+\frac{\sqrt{2}r\tau (bc-ad)}{b}sin(\theta+\frac{\pi}{4}) -\frac{\sqrt{2}\tau r_{0}(bc-ad)}{b}sin(\theta_{0}+\frac{\pi}{4})\right] }.
		\end{split}
	\end{align}	
	Let $n=m$, we get
	\begin{align}
		\begin{split}
			&\;\quad f\left(\mathbf{t}-\mathbf{t_{0}} \right)\\ 
			&=\left(-i \right)^{\frac{n}{2}} 2\pi Cw_{3}\int_{0}^{+\infty}\sum_{n=-\infty}^{+\infty}J_{n}\left(\frac{r\rho}{b}\right)J_{n}\left(\frac{r_{0}\rho}{b}\right)J_{n}\left(\frac{\sqrt{2}\rho (b\eta-d\tau)}{b}\right) \\
			&\times e^{i\left[ -
				\frac{d}{2b}\rho^{2} 	
				-\frac{a}{2b}\left( r^{2}+r_{0}^{2}\right) +\frac{arr_{0}}{b}cos(\theta-\theta_{0})+\frac{\sqrt{2}r\tau (bc-ad)}{b}sin(\theta+\frac{\pi}{4}) -\frac{\sqrt{2}\tau r_{0}(bc-ad)}{b}sin(\theta_{0}+\frac{\pi}{4})\right] } \\
			&\times H_{n}^{A}\left[ f_{n}\right] \left(\rho \right)e^{in\left( \theta-\theta_{0}\right) }\rho\rm{d}\rho.
		\end{split}
	\end{align}	
	Change the integral and summation signs, we have
	\begin{align}
		\begin{split}
			&\;\quad f\left(\mathbf{t}-\mathbf{t_{0}} \right)\\ &=\left(-i \right)^{\frac{n}{2}} 2\pi Cw_{3}\sum_{n=-\infty}^{+\infty}e^{in\left( \theta-\theta_{0}\right)}\\
			&\times e^{i\left[ -
				\frac{d}{2b}\rho^{2} 	
				-\frac{a}{2b}\left( r^{2}+r_{0}^{2}\right) +\frac{arr_{0}}{b}cos(\theta-\theta_{0})+\frac{\sqrt{2}r\tau (bc-ad)}{b}sin(\theta+\frac{\pi}{4}) -\frac{\sqrt{2}\tau r_{0}(bc-ad)}{b}sin(\theta_{0}+\frac{\pi}{4}) \right] }\\
			& \times \int_{0}^{+\infty}H_{n}^{A}\left[ f_{n}\right] \left(\rho \right) J_{n}\left(\frac{r\rho}{b}\right)J_{n}\left(\frac{r_{0}\rho}{b}\right)J_{n}\left(\frac{\sqrt{2}\rho (b\eta-d\tau)}{b}\right)\rho\rm{d}\rho. 
		\end{split}
	\end{align}	
	The proof is completed.  \end{proof}

\subsection{Convolution theorem}
Let us give the convolution concept of the 2D OLCT in polar coordinates as follows:
\begin{defn} 
	The convolution operation $\ast^{A}$ of the 2D OLCT in polar coordinates for two function $f\left( r,\theta\right) $ and $g\left( r,\theta\right)$ is defined by   
	\begin{align}
		\begin{split}
			\left(f\ast^{A}g \right) \left(r,\theta \right)  &=\frac{e^{-i\left[ \frac{a}{2b}r^{2}+\frac{\sqrt{2}\tau r}{b} sin\left(\theta+\frac{\pi}{4} \right)\right]  }}{2\pi}\int_{0}^{+\infty}\int_{0}^{2\pi}e^{i\left[ \frac{a}{2b}r_{0}^{2}+\frac{\sqrt{2}\tau r_{0}}{b} sin\left(\theta_{0}+\frac{\pi}{4} \right)\right]  }\\
			&\times F^{-1}\left[ e^{-i\rho r_{0}cos\left(\phi-\theta_{0} \right) }F^{A}\left(b\rho,\phi \right) \right](r,\theta) r_{0}\rm{d}r_{0}\rm{d}\theta_{0}.
		\end{split}
	\end{align}	
\end{defn}

In the following, we obtain the convolution theorem of the 2D OLCT in polar cooddinates.
\begin{thm}[Convolution theorem]
	Let $z\left(r,\theta \right)=\left(f\ast^{A}g \right) \left(r,\theta \right)$. Then, there is a relation
	\begin{align}
		\begin{split}
			Z^{A}\left( \rho,\phi\right) =F^{A}\left( \rho,\phi\right)G^{A}\left( \rho,\phi\right),
		\end{split}
	\end{align}	
	where $Z^{A}\left( \rho,\phi\right)$, $F^{A}\left( \rho,\phi\right)$ and $G^{A}\left( \rho,\phi\right)$ denote the 2D OLCT in polar coordinates of the functions $z(r,\theta)$, $f(r,\theta)$ and $g(r,\theta)$, respectively.	 
\end{thm}
\begin{proof}
	In view of (39), we get
	\begin{align}
		\begin{split}
			e^{-i\rho r_{0} cos\left( \phi-\theta_{0} \right)}=\sum_{m=-\infty}^{+\infty}i^{-m}J_{m}\left( \rho r_{0} \right)e^{-im\theta_{0}}e^{im\phi}.
		\end{split}
	\end{align}
	By (12) and (58), it follows that 
	\begin{align}
		\begin{split}
			&\;\quad F^{-1}\left[ e^{-i\rho r_{0}cos\left(\phi-\theta_{0} \right) }F^{A}\left(b\rho,\phi \right) \right](r,\theta)\\
			&=\frac{1}{2\pi}\int_{0}^{+\infty}\int_{0}^{2\pi}\sum_{m=-\infty}^{+\infty}i^{-m}J_{m}\left( \rho r_{0} \right)e^{-im\theta_{0}}e^{im\phi}\\
			&\times\sum_{n=-\infty}^{+\infty}H_{n}^{A}\left[f_{n} \right]\left( b\rho \right)e^{in\phi}\sum_{k=-\infty}^{+\infty}i^{k}J_{k}\left( r\rho \right)e^{ik\theta}e^{-ik\phi}\rho\rm{d}\rho\rm{d}\phi. \\
			&=\frac{1}{2\pi}\int_{0}^{+\infty}\sum_{m=-\infty}^{+\infty}i^{-m}J_{m}\left( \rho r_{0} \right)e^{-im\theta_{0}}\sum_{n=-\infty}^{+\infty}H_{n}^{A}\left[f_{n} \right]\left( b\rho \right)\\
			&\times\sum_{k=-\infty}^{+\infty}i^{k}J_{k}\left( r\rho \right)e^{ik\theta}\int_{0}^{2\pi}e^{im\phi+in\phi-ik\phi}\rm{d}\phi\rho\rm{d}\rho.
		\end{split}
	\end{align}
	Since 
	\begin{align}
		\begin{split}
			\int_{0}^{2\pi}e^{im\phi+in\phi-ik\phi}\rm{d}\phi&=2\pi\delta_{m+n-k}=\left \{  
			\begin{array}{ll}
				2\pi,                         & m+n-k=0\\
				0,                            & otherwise
			\end{array} \right.
		\end{split},
	\end{align}
	where $\delta_{\cdotp,\cdotp}$ denotes the kronecker delta operator.\\
	Hence, (59) can be written as
	\begin{align}
		\begin{split}
			i^{n}\int_{0}^{+\infty}\sum_{m=-\infty}^{+\infty}i^{-m}J_{m}\left( \rho r_{0} \right)e^{-im\theta_{0}}\sum_{k=-\infty}^{+\infty}H_{k-m}^{A}\left[f_{k-m} \right]\left( b\rho \right)J_{k}\left( r\rho \right)e^{ik\theta}\rho\rm{d}\rho.
		\end{split}
	\end{align}
	We make the following symbols 
	\begin{align}
		\begin{split}
			\tilde{z}(r,\theta)=e^{i\left[ \frac{a}{2b}r^{2}+\frac{\sqrt{2}\tau r}{b}sin(\theta+\frac{\pi}{4})\right] }z\left(r,\theta \right),
		\end{split}
	\end{align}
	\begin{align}
		\begin{split}		
			\tilde{g}(r_{0},\theta_{0})=e^{i\left[ \frac{a}{2b}r_{0}^{2}+\frac{\sqrt{2}\tau r_{0}}{b}sin(\theta_{0}+\frac{\pi}{4})\right] }g\left(r_{0},\theta_{0} \right),
		\end{split}
	\end{align}	
and the Fourier series expansion of the 2D OLCT has a form	 
	\begin{align}
		\begin{split}
			\tilde{g}(r_{0},\theta_{0})=\sum_{n=-\infty}^{+\infty}\tilde{g}_{n}\left( r_{0} \right)e^{in\theta_{0}}, 
		\end{split}
	\end{align}
	where $\tilde{g}_{n}\left( r_{0} \right)=w_{2}e^{i\frac{a}{2b}r_{0}^{2}}g_{n}\left( r_{0} \right)$, $w_{2}$ is given by (15) and $g_{n}\left( r_{0} \right)$ is the $n$th term in the Fourier series for the function $g(r_{0},\theta_{0})$.\\
	Using (56), (61) and (64) that
	\begin{align}
		\begin{split} 
			\tilde{z}(r,\theta)&=i^{n}\frac{1}{2\pi}\int_{0}^{+\infty}\int_{0}^{+\infty}\sum_{n=-\infty}^{+\infty}\tilde{g}_{n}\left( r_{0} \right)\sum_{m=-\infty}^{+\infty}i^{-m}J_{m}\left( \rho r_{0} \right)\\
			&\times\sum_{k=-\infty}^{+\infty}H_{k-m}^{A}\left[f_{k-m} \right]\left( b\rho \right)J_{k}\left( r\rho \right)e^{ik\theta}\int_{0}^{2\pi}e^{in\theta_{0}-im\theta_{0}}r_{0}\rm{d} r_{0}\rho\rm{d}\rho.\\
		\end{split}
	\end{align}
	Let $m=n$ and by (30), we get
	\begin{align}
		\begin{split} 
			\tilde{z}(r,\theta)&=\int_{0}^{+\infty}\sum_{m=-\infty}^{+\infty}\int_{0}^{+\infty}\tilde{g}_{m}\left( r_{0} \right)J_{m}\left( \rho r_{0} \right)\\
			&\times\sum_{k=-\infty}^{+\infty}H_{k-m}^{A}\left[f_{k-m} \right]\left( b\rho \right)J_{k}\left( r\rho \right)e^{ik\theta}\rho\rm{d}\rho\\
			&=\sum_{k=-\infty}^{+\infty}\int_{0}^{+\infty}\sum_{m=-\infty}^{+\infty}H_{m}\left[\tilde{g}_{m}\left( r_{0} \right) \right]\left(\rho \right)\\
			&\times H_{k-m}^{A}\left[f_{k-m} \right]\left( b\rho \right)J_{k}\left( r\rho \right)e^{ik\theta}\rho\rm{d}\rho.   
		\end{split}
	\end{align}
	The (66) is a Fourier series expansion, then the $k$th term is
	\begin{align}
		\begin{split} 
			\tilde{z}_{k}(r)&=\int_{0}^{+\infty}\sum_{m=-\infty}^{+\infty}H_{m}\left[\tilde{g}_{m}\left( r_{0} \right) \right]\left(\rho \right)H_{k-m}^{A}\left[f_{k-m} \right]\left( b\rho \right)J_{k}\left( r\rho \right)\rho\rm{d}\rho, 
		\end{split}
	\end{align}
	where $\tilde{z}_{k}\left( r \right)=w_{2}e^{i\frac{a}{2b}r^{2}}g_{k}\left( r_{0} \right)$, $w_{2}$ is given by (15) and $z_{k}\left( r \right)$ is the $k$th term in the Fourier series for the function $z(r,\theta)$.\\
	Apply the inverse HT formula, (67) becomes
	\begin{align}
		\begin{split} 
			H_{k}\left[ \tilde{z}_{k}\right]\left( \rho\right)  &=\sum_{m=-\infty}^{+\infty}H_{m}\left[\tilde{g}_{m} \right]\left(\rho \right)H_{k-m}^{A}\left[f_{k-m} \right]\left( b\rho \right). 
		\end{split}
	\end{align}
	According to the relationship between the OLCHT and the HT, we obtain
	\begin{align}
		\begin{split} 
			H_{k}^{A}\left[ z_{k}\right]\left( b\rho\right)  &=i^{k}\frac{w_{1}\ell_{A}}{b}e^{i\frac{d}{2b}\rho^{2}}H_{k}\left[ \tilde{z}_{k}\right] \left( \rho\right),\\ 
			H_{m}^{A}\left[ g_{m}\right]\left( b\rho \right)  &=i^{m}\frac{w_{1}\ell_{A}}{b}e^{i\frac{d}{2b}\rho^{2}}H_{m}\left[ \tilde{g}_{m}\right]  \left( \rho\right),
		\end{split}
	\end{align}
	where $w_{1}$ is given by (15) and $\ell_{A}$ is given by (6). Hence 
	\begin{align}
		\begin{split} 
			i^{-k}H_{k}^{A}\left[ z_{k}\right]\left( \rho\right)  &=i^{-m}\sum_{m=-\infty}^{+\infty}H_{m}^{A}\left[ \tilde{g}_{m}\right]\left( \rho \right)H_{k-m}^{A}\left[f_{k-m} \right]\left( \rho \right).
		\end{split}
	\end{align}
	By (60), we have
	\begin{align}
		\begin{split} 
			H_{k}^{A}\left[ z_{k}\right]\left( \rho\right)  &=\sum_{m=-\infty}^{+\infty}H_{m}^{A}\left[ g_{m}\right]\left( \rho \right)H_{k-m}^{A}\left[f_{k-m} \right]\left( \rho \right),
		\end{split}
	\end{align}
	that is, a convolution relation
	\begin{align}
		\begin{split} 
			Z_{k}^{A}\left( \rho\right)  &=\sum_{m=-\infty}^{+\infty}G_{m}^{A}\left( \rho \right)F_{k-m}^{A}\left( \rho \right).
		\end{split}
	\end{align}

	According to the formula (46), as we all know, the convolution of two sets of the Fourier coefficients is equivalent to the multiplication of functions \cite{[32]}. Therefore, it can be seen from (72) that there is a multiplicative relationship
	\begin{align}
		\begin{split}
			Z^{A}\left( \rho,\phi\right) =F^{A}\left( \rho,\phi\right)G^{A}\left( \rho,\phi\right),
		\end{split}
	\end{align}	
	where $Z^{A}\left( \rho,\phi\right)$, $F^{A}\left( \rho,\phi\right)$ and $G^{A}\left( \rho,\phi\right)$ denote the 2D OLCT in polar coordinates of the functions $z(r,\theta)$, $f(r,\theta)$ and $g(r,\theta)$, respectively. \end{proof}

\section{Conclusions}
\label{Con}
According to the definition of the OLCT in the two-dimensional Cartesian coordinate system, this paper obtains its new expression in polar coordinates, and derives the mathematical formula of the OLCHT. Then, the relationship between the 2D OLCT and the OLCHT is derived. Finally, using this relationship, the spatial shift theorem and convolution theorem of the 2D OLCT are proved. These contents are new achievements in the field of the 2D OLCT polar coordinates. In further work, we will consider the sampling theorem of the 2D OLCT in optics and signal processing.\

\bibliographystyle{elsarticle-num}

\section*{Declaration of competing interest}
The authors declare that they have no known competing financial interests or
personal relationships that could have appeared to influence the work reported
in this paper.
\section*{Acknowledgments}
This work was supported by the National Natural Science Foundation of China [No. 61671063].

\end{document}